\documentclass[
submission
]{dmtcsepisciences}

\usepackage[utf8]{inputenc}
\usepackage{subfigure}

\usepackage{amssymb}

\usepackage{amsthm}
\usepackage{graphicx}
\usepackage{amsmath}
\usepackage{amsfonts}
\usepackage{hyperref}
\usepackage{xcolor}
\usepackage{dsfont}
\usepackage{scalerel}

\usepackage{pgf,tikz,calc}
\usetikzlibrary{math}

\usetikzlibrary{arrows}
\usetikzlibrary{shapes}
\usepackage{multirow}
\usepackage{latexsym}
\usepackage{setspace}
\PassOptionsToPackage{boxed,section}{algorithm}
\usepackage{algorithm}
\usepackage{algorithmic}
\usepackage{enumerate}
\usepackage{amsmath}			
\usepackage{amssymb}			
\usepackage{url}
\usepackage{setspace}
\usetikzlibrary{arrows}
\usetikzlibrary{shapes}
\usetikzlibrary{decorations.markings}
\usepackage{geometry}
\usepackage{bbm}

\newtheorem{theorem}{\em Theorem}
\newtheorem{proposition}[theorem]{\em Proposition}

\newtheorem{lemma}[theorem]{\em Lemma}
\newtheorem{remark}[theorem]{\em Remark}
\newtheorem{corollary}[theorem]{\em Corollary}

\newtheorem{observation}[theorem]{\em Observation}

\usepackage[round]{natbib}

\author{Hanna Furma\'nczyk\affiliationmark{1}\thanks{hanna.furmanczyk@ug.edu.pl}
  \and Vahan Mkrtchyan\affiliationmark{2}\thanks{vahan.mkrtchyan@gssi.it}
  }
\title{Graph theoretic and algorithmic aspect of the equitable coloring problem in block graphs\footnote{The work of the second author has been partially supported by the Italian MIUR PRIN 2017 Project ALGADIMAR ``Algorithms, Games, and Digital Markets.''}}
\affiliation{
  Institute of Informatics, Faculty of Mathematics, Physics and Informatics, University of Gda\'nsk,  Gda\'nsk, Poland\\
  Gran Sasso Science Institute, L'Aquila, Italy}
\received{2020-10-27}

\revised{2021-09-30, 2022-03-29, 2022-09-28}

\accepted{2022-10-15}

\begin{document}
\publicationdetails{23}{2022}{2}{7}{6860}
\maketitle
\begin{abstract}
  An equitable coloring of a graph $G$ is its (proper) vertex-coloring such that the sizes of any two color classes differ by at most one. In this paper, we consider the equitable coloring problem in block graphs. Recall that the latter are graphs in which each 2-connected component is a complete graph. The problem of equitable coloring remains hard in the class of block graphs. In this paper, we present some graph theoretic results relating various parameters. We use them in order to trace some algorithmic implications, mainly dealing with the fixed-parameter tractability of the problem.
\end{abstract}

\section{Introduction}
\label{sec:intro}

\subsection{Some graph theory concepts}
\label{subsec:GraphTheoryConcepts}
In this paper, we consider finite undirected graphs. They do not contain loops or parallel edges.  We refer to \cite{harary} for non-defined concepts on graphs.

Two vertices of a graph $G$ are \emph{independent} 
if and only if there is no edge joining them. A set of vertices is independent, if its vertices are pairwise independent. Let $\alpha(G)$ be the cardinality of a largest independent set in a graph $G$. Similarly, two edges of a graph are independent, if they do not share a vertex. A \emph{matching} is a subset of edges of a graph such that any two edges in it are independent. Let $\nu(G)$ denote the size of a largest matching of $G$. A matching is \emph{perfect} if it covers all the vertices of the graph. A \emph{vertex cover} is a subset of vertices whose removal results in a graph with no edge. The size of a smallest vertex cover is denoted by $\tau(G)$. In any graph $G$, we have
\[\nu(G)\leq \tau(G)\leq 2\nu(G).\]
A \emph{clique} of a graph $G$ is a complete subgraph of $G$. For a graph $G$, let $\omega(G)$ denote the size of a largest clique of $G$. A clique $Q$ is \emph{maximal} in $G$ if and only if there is no other clique in the graph containing $Q$. 
A vertex $v$ is a \emph{cut-vertex}, if $G-v$ contains more connected components than $G$. The number of cut-vertices in a graph $G$ is denoted by $c(G)$. The \emph{line graph} of an undirected simple graph $G$ is a graph $L(G)$ obtained by associating a vertex of $L(G)$ with each edge of $G$ and connecting two vertices in $L(G)$ with an edge if and only if the corresponding edges of $G$ have a vertex in common.
If $G$ is a connected graph then let $d(u,v)$ denote the length of a shortest path connecting the vertices $u$ and $v$. For a vertex $u$, its \emph{eccentricity}, denoted by $\epsilon_G(u)$, is defined as $\max_{v\in V}\{d(u,v)\}$. The \emph{radius} of $G$, denoted by $rad(G)$, is $\min_{v\in V}\{\epsilon_G(v)\}$ and its \emph{diameter}, $diam(G)$, is defined as $\max_{v\in V}\{\epsilon_G(v)\}$. The \emph{center} of a graph is the subset of vertices whose eccentricity is equal to the radius of the graph. For any graph $G$, we have
\[rad(G)\leq diam(G)\leq 2\cdot rad(G).\]
Following \cite{gomes:structural}, we define a \emph{cluster graph} as a graph formed from the disjoint union of complete graphs. For a given graph $G$, its \emph{distance to the cluster}, denoted by $dc(G)$, is the smallest number of vertices of $G$, whose removal results in a cluster graph. A set $D$ is called a $dc$-set, if $|D|=dc(G)$ and $G-D$ is a cluster.

A \emph{block} of a graph $G$ is a maximal 2-connected subgraph of $G$.
A graph $G$ is a \emph{block graph}, if each of its blocks is a clique. If $G$ is a block graph, then a vertex is \emph{simplicial} if and only if it is not a cut-vertex. Clearly, the neighbors of a simplicial vertex are in the same clique. A maximal clique in a block graph is \emph{pendant} if and only if it contains exactly one cut-vertex of $G$. Let $p(G)$ be the number of pendant cliques of $G$ while $s(G)$ be the number of simplicial vertices of $G$. Clearly, for any block graph $G$, we have $p(G)\leq s(G)$. A block graph $G$, with at least two blocks, is called a \emph{star of cliques} or a \emph{clique-star}, if $G$ contains a vertex that lies in all cliques of $G$. Observe that this vertex should be the unique cut-vertex of $G$. 
For a vertex $v$ of a block graph $G$, the \emph{clique-degree} of $v$ is the number of cliques of $G$ containing $v$. If $T$ is a clique of $G$, then the \emph{clique-degree} of $T$ is the number of cliques of $G$ that differ from $T$ and intersect $T$. A connected block graph is called a \emph{path of cliques}, if each vertex in it has clique-degree at most two. Alternatively, one can say that a connected block graph is a \emph{path of cliques}, if each clique in it has clique-degree at most two.

We will assign natural numbers to the maximal cliques of a block graph $G$. This number will be called the \emph{level} of a clique. We do it by the following algorithm: all pendant cliques of $G$ are assigned level 1. Then we remove all simplicial vertices of all pendant cliques of $G$ in order to obtain the block graph $G_1$. All pendant cliques of $G_1$ get level 2 in $G$. Then, we remove all simplicial vertices of all pendant cliques of $G_1$ in order to obtain the block graph $G_2$. Then we repeat this process until all blocks of $G$ get their levels. Finally, if we are left with a singleton, we do not assign a level to it. Observe that the star of cliques are exactly those connected block graphs which do not contain blocks of level at least 2.

\subsection{Introduction to some graph problems}
This paper deals with a variant of classical \textsc{Vertex Coloring} problem, namely \textsc{Equitable Coloring}.  
If the set of vertices of a graph $G$ can be partitioned into $k$ classes 
$V_1, V_2, \ldots,V_k$ such that each $V_i$ is an independent set and the condition $||V_i|-|V_j||\leq 1$ holds for every pair ($i, j$) then $G$ is said to be {\it equitably k-colorable}. The smallest integer $k$ for which $G$ is equitably $k$-colorable is known as the {\it equitable chromatic number} of $G$ and it is denoted by $\chi_{=}(G)$. 

This model was introduced by 
\cite{meyer} and it has attracted attention of many graph theory specialists for almost 50 years. The conducted studies are mainly focused on the proving of known conjectures for particular graph classes (cf. \cite{chenD, kierstead, kostochkadeg, kostochkadeg2}), analysis of the problem’s complexity (cf. \cite{tight}), designing exact algorithms for polynomial cases (cf. \cite{fast}), and approximation algorithms or heuristics for hard cases (cf. \cite{fjk, dsatur}). We know that the \textsc{Equitable Coloring} problem is NP-complete in general case, as a particular case of \textsc{Vertex Coloring}.

{Note that \textsc{Bin Packing Problem with Conflicts} (BPC) is closely related to \textsc{Equitable Coloring}. BPC is defined as follows. We are given a set $V$ of $n$ items of weights $w_1, w_2, \ldots, w_n$, and $k$ identical bins of capacity $c$. Two items $i$ and
$j$ are said to be \emph{conflicting} if and only if they cannot be assigned to the same bin. The problem is
to assign all items in the least possible number of bins while ensuring that the total weight of all
items assigned to a bin does not exceed $c$ and that no bin contains conflicting items. Note that the problem with $c=n/k$ and weights equal to 1 is equivalent to an equitable coloring of the corresponding conflict graph. Some exemplary heuristics for solving \textsc{BPC} can be found in \cite{bpc:2004, bpc:alg}.}

An interesting overview of the results of studies over equitable coloring can be found in \cite{furm:en_book} and \cite{lih}. This issue is very important due to its many applications (creating timetables, task scheduling, transport problems, networks, etc.) (see for example \cite{furm:4sch, furm:3sch}). 

\subsection{Introduction to parameterized complexity}
Very recently a few papers investigating the parameterized complexity of \textsc{Equitable Coloring} have been published (cf. \cite{iterated, fellows, fiala, gomes:structural, EqParamFPT}). Recall that if $\Pi$ is an algorithmic problem and $t$ is a parameter, then the pair $(\Pi, t)$ is called a \emph{parameterized} problem. The parameterized problem $(\Pi, t)$ is \emph{fixed-parameter tractable} (or $\Pi$ is \emph{fixed-parameter tractable with respect to the parameter} $t$) if there is an algorithm $A$ that solves $\Pi$ exactly, whose running-time is $g(t)\cdot poly(size)$. Here $g$ is some (computable) function of $t$, $size$ is the length of the input and $poly$ is a polynomial function. Usually, such an algorithm $A$ is called an FPT algorithm for $(\Pi, t)$. A (parameterized) problem is called \emph{para}NP-\emph{hard}, if it remains NP-hard even when the parameter under consideration is constant. In the classical complexity theory, there is the notion of NP-hardness that indicates that a certain problem is unlikely to be polynomial time solvable. It relies on the assumption P$\neq$NP. The classical \textsc{Satisfiability} problem is an NP-hard problem and any problem such that \textsc{Satisfiability} can be reduced to it is NP-hard, too. Similarly, in parameterized complexity theory there is the notion of W[1]-\emph{hardness}, which indicates that a certain parameterized problem is unlikely to be fixed-parameter tractable. It relies on the assumption FPT$\neq$W[1] which says that not all problems from W[1] are fixed-parameter tractable. The \textsc{Clique} problem where the parameter under consideration is $k$ - the size of the clique - is an example of a W[1]-hard problem, and any problem such that \textsc{Clique} with respect to $k$ can be FPT-reduced to it, is W[1]-hard, too. Recall that an \emph{FPT reduction} between two parameterized problems $(\Pi_1, t_1)$ and $(\Pi_2, t_2)$ is an algorithm $R$ that maps instances of $\Pi_1$ to those of $\Pi_2$, such that (a) for any instance $I_1\in \Pi_1$, we have $I_1$ is a yes-instance of $\Pi_1$ if and only if $R(I_1)$ is a yes-instance of $\Pi_2$, (b) there is a computable function $h$ such that for any instance $I_1\in \Pi_1$ $t_2(R(I_1))\leq h(t_1(I_1))$, (c) there is a computable function $g$ such that $R$ runs in time $g(t_1)\cdot poly(size)$. The reader can learn more about this topic from \cite{FPTbook}, that can be a good guide for algorithmic concepts that are not defined in this paper.

\cite{fellows} showed that the \textsc{Equitable Coloring} problem is W[1]-hard, parameterized by the treewidth plus the number of colors. 
\cite{fiala} considered another structural parameter - vertex cover. They showed that the problem is FPT with respect to it. 
\cite{gomes:structural} established new results for some other parameters: W[1]-hardness for pathwidth and feedback vertex set, and fixed parameter tractability for distance to cluster and co-cluster, as well as distance to disjoint paths of bounded length. In the same paper 
the authors consider also kernelization for the problem of \textsc{Equitable Coloring}. They presented a linear kernel for the distance to clique parameter and a cubic
kernel when parameterized by the maximum leaf number. In the second paper, 
\cite{EqParamFPT} considered parameterized complexity of \textsc{Equitable Coloring} problem for subclasses of perfect graphs. They showed W[1]-hardness for block graphs when parameterized by the number of colors, and for $K_{1,4}$-free interval graphs when parameterized by treewidth, number of colors and maximum degree. 

\subsection{Outline of the paper}
In this paper, we further the study of \textsc{Equitable Coloring} on block graphs. As shown in \cite{EqParamFPT}, this class is a non-trivial subclass of chordal and perfect graphs. For block graphs, it is shown in \cite{EqParamFPT} that the problem is W[1]-hard with respect to the treewidth, diameter and the number of colors. In particular, it means that under the standard assumption FPT$\neq$W[1] in parameterized complexity theory, the problem is unlikely to be polynomial time solvable in block graphs. In this paper, we investigate parameterized complexity of \textsc{Equitable Coloring} of block graphs with respect to many other parameters thus completing the state of art in this area.
The paper is organized as follows. We start with simple observations followed from the literature that affect our research. In Section \ref{sec:spanning}, we investigate the problem with respect to some parameters that are related to minimum and maximum number of leaves in a spanning tree of a graph. In the following section we prove that the \textsc{Equitable Coloring} is FPT with respect to the domination number for block graphs showing at the same time that the problem with this parameter is much easier for block graphs than for general graphs. In Section \ref{sec:independence}, we investigate the problem with respect to some parameters that are related to the independence of vertices and edges of graphs. In Section \ref{sec:otherparameters}, we consider other parameters and relate them in block graphs. Finally, we conclude the paper by presenting some open problems that we feel deserve to be investigated. 

\section{Some observations}
Before we start presenting our results, we list some observations and corollaries.

 \begin{lemma}[\cite{Sasak:2010}]
	\label{lem:Redparam} Let $\Pi$ be an algorithmic problem, and let $k_1$ and $k_2$ be some parameters. Assume that there is a (computable) function $g:\mathbb{N}\rightarrow \mathbb{N}$ such that for any instance $I$ of $\Pi$, we have $k_1(I)\leq g(k_2(I))$. Then, if $\Pi$ is FPT with respect to $k_1$, it is also FPT with respect to $k_2$.
\end{lemma}

\begin{theorem}[\cite{EqParamFPT}]
\label{thm:gomes:diam4}
EQUITABLE COLORING of block graphs of diameter
at least four parameterized by the number of colors and treewidth is $W[1]$-hard.\label{diam4}
\end{theorem}

\begin{theorem}[\cite{gomes:structural}]
EQUITABLE COLORING is FPT when parameterized
by the distance to cluster. \label{thm:gomes:dc}
\end{theorem}

\begin{theorem}[\cite{fiala}]
EQUITABLE COLORING is FPT when parameterized by vertex cover.\label{thm:fiala:cover}
\end{theorem}

\begin{theorem}[\cite{fellows}]
EQUITABLE COLORING is $W[1]$-hard, parameterized by treewidth.
\end{theorem}

\begin{theorem}[\cite{EqParamFPT}]
EQUITABLE COLORING is FPT when parameterized
by the treewidth of the complement graph.
\end{theorem}

Directly, we have
\begin{corollary}
EQUITABLE COLORING of complements of block graphs with fixed clique number is polynomialy solvable.
\end{corollary}

\section{Equitable coloring and number of leaves in a spanning tree}
\label{sec:spanning}

In this section, we consider block graphs and the \textsc{Equitable Coloring} problem from the perspective of the number of leaves in a spanning tree, their minimum and maximum number.

Let $MinLeaf(G)$ ($MaxLeaf(G)$) be the smallest (largest) number of leaves in any spanning tree of $G$. $MaxLeaf$ was considered 
independently in \cite{eq_partition} and \cite{gomes:structural}. In particular, \cite{eq_partition} shows that \textsc{Equitable Coloring} is FPT with respect to $MaxLeaf$ in general (not necessarily block) graphs. Note that these two parameters, $MinLeaf$ and $MaxLeaf$, are NP-hard to compute in arbitrary graphs. Below we present two observations that imply that these two parameters can be easily computed in the class of block graphs.

\begin{proposition}Let $G$ be a connected block graph. Then $MinLeaf(G)$ coincides with the number of pendant cliques in $G$.
\end{proposition}

\begin{proof}
Observe that any spanning tree of $G$ has at least one degree-one vertex in a pendant clique of $G$.  Thus, $MinLeaf(G)\geq p(G)$.  Moreover, any simplicial vertex of a pendant clique can be made as a leaf in the spanning tree with the smallest number of leaves.

In order to show the converse inequality, let us note that we can pick a Hamiltonian path in each non-pendant clique that begins and ends in a cut-vertex and a Hamiltonian path on each pendant clique that begins with its cut-vertex. When we join all these Hamiltonian paths, we obtain a spanning tree of $G$ that has one leaf per pendant clique. The proof is complete.
\end{proof}
\begin{proposition}
Let $G$ be a connected block graph. Then $MaxLeaf(G)$ coincides with the number of simplicial vertices in $G$.\label{prop:maxleaf}
\end{proposition}
\begin{proof} 
First of all, observe that no cut-vertex of $G$ can be a leaf in a spanning tree of $G$. Hence we get $MaxLeaf(G)\leq s(G)$. 
Thus, in order to complete the proof of our proposition, it suffices to show that any connected block graph has a spanning tree whose all leaves are the simplicial vertices of $G$.
To do this, take a simple path in each non-pendant clique that begins and ends in a cut 
vertex and omits all its simplicial vertices. Now, join the paths in order to obtain a spanning tree for all non-simplicial vertices. Next, add 
each simplicial vertex to one of its neighbors visited in the first step of the algorithm. Note that we obtain a spanning tree of $G$. The proof is complete.
\end{proof}

Since $MaxLeaf(G)$ coincides with the number of simplicial vertices for block graph $G$ then $|V|-MaxLeaf(G)$ is equal to the number of cut-vertices. Note that removing all cut-vertices from a block graph $G$ leads to a union of cliques. Thus,
$$dc(G)\leq |V|-MaxLeaf(G) \leq |V|-MinLeaf(G).$$

Due to Theorem \ref{thm:gomes:dc} and Lemma \ref{lem:Redparam} we have the following results.
\begin{proposition}
\textsc{Equitable Coloring} in block graphs is FPT with respect to $|V|-MaxLeaf(G)$. \label{prop:n-max}
\end{proposition}
\begin{proposition}
\textsc{Equitable Coloring} in block graphs is FPT with respect to $|V|-MinLeaf(G)$. 
\end{proposition}

\section{Equitable coloring and domination number of block graphs}

Let us recall a result by 
\cite{nieminen}, providing the relation between the number of leaves in a maximum spanning forest of graph $G$ and its  domination number. Recall that a \emph{spanning forest} is a subgraph of $G$ which is a forest
and has the same vertex-set as that of $G$. A spanning forest $F$ of $G$ is called \emph{maximum}
if it has the largest possible number of  pendant edges among all spanning forests of $G$.
This number is denoted by $Leaf_F(G)$. 
A \emph{dominating set} for a graph $G = (V, E)$ is a subset $D$ of $V$ such that every vertex not in $D$ is adjacent to at least one member of $D$. The \emph{domination number} $\gamma(G)$ is the number of vertices in a smallest dominating set for $G$.

\begin{theorem}[\cite{nieminen}]
Let $G$ be a simple graph. Then $\gamma(G)+Leaf_F(G)=|V(G)|$. 
\end{theorem}

Note that $MaxLeaf(G)\leq Leaf_F(G)$. Thus, we have $\gamma(G)=|V(G)|-Leaf_F(G) \leq |V(G)|- 
MaxLeaf(G)$. 
When $G$ is a block graph, we have that the number of cut-vertices, denoted by $c(G)$, is equal to $|V(G)| - MaxLeaf(G)$. So, we have $\gamma(G)\leq c(G)$ for block graphs. Moreover, we prove the following

\begin{proposition}
Let $G=(V,E)$ be a (not necessarily block) graph. Then $c(G)\leq 2 \gamma(G)$.
\end{proposition}
\begin{proof} First of all, observe that proving this inequality in arbitrary graphs is equivalent to its restriction in block graphs. In order to see this, let $G$ be any graph. Consider a graph $H$ obtained from $G$ by making each block $B$ of $G$ a clique by adding edges to $B$. No edge of $H$ joins two vertices that lie in different blocks of $G$ unless one of them is a cut vertex that lies in both blocks.

Observe that $H$ is a block graph. Moreover, $c(H)=c(G)$ and $\gamma(H)\leq \gamma(G)$. Hence,
\[c(G)=c(H)\leq 2\gamma(H)\leq 2 \gamma(G).\]
Thus, w.l.o.g. we can assume that the graph $G$ is a block graph. Moreover, we can assume that $G$ is connected. 
We will prove the statement by induction on $c(G)$. Note that, if $0\leq c(G)  \leq 1$, the domination number is also equal to 1, so the statement holds in both cases. Now, let us assume that our statement holds for every block graph $G$ with $c(G)<k$, $k\geq 2$. Now, let $G$ be a block graph with $k$ cut-vertices. We choose a clique $Q$ being a clique of level 2 in $G$. Such a clique exists because $G$ has at least 2 cut-vertices. Note that vertices of $Q$ are simplicial or they are included in cliques of level 1, excluding at most one vertex, let us say $z$. It can be contained in cliques of higher level. Let $x$ be a cut-vertex of $Q$ contained in at least one clique of level 1. If $Q$ contains at least three cut-vertices: $x$, $z$, and let us say $y$, then note that $y$ is contained in at least one clique of level 1, and certainly vertex $y$ can be needed to dominate vertices of $G$, similarly to vertex $x$. 
Let $H$ be a block graph obtained from $G$ by deleting all pendant cliques containing $x$, together with $x$.
Note that if more than one clique containing $x$ were removed to obtain graph $H$, vertex $x$ is included in any dominating set of $G$ of minimum size. Also in the case where there was only one clique of level 1 in $G$ containing $x$, let us name it $J$, exactly one vertex of $J$ must belong to a dominating set of $G$ of size $\gamma(G)$ and we can assume that it is vertex $x$. 
So, we have $\gamma(H)=\gamma(G)-1$, while the number of cut-vertices in $H$ was decreased by 1. Thus, using the induction assumption for $H$, we have
$$c(G)-1=c(H)\leq 2\gamma(H)=2(\gamma(G)-1).$$
Hence,
$$c(G)\leq 2\gamma(G)-1\leq 2\gamma(G).$$ 

So, we can assume that all cliques of level 2 contain exactly two cut-vertices: $x$ and $z$. Note that $x$ dominates all vertices of pendant cliques containing $x$ as well as the vertices of $Q$, i.e. $z$ and all simplicial vertices of $Q$. 
Let us assume that the vertex $z$ is included in exactly one clique, excluding $Q$. Then, let $Q^*$ be the clique containing $z$ different from $Q$. Note that any vertex of $Q^*$ dominates the other vertices of $Q^*$. Thus, we can assume that we can choose any vertex of $Q^*$, excluding $z$, to a dominating set of $G$ of size $\gamma(G)$. So, let $H$ be the graph obtained from $G$ by deleting $Q$, including $z$, and all pendant cliques containing $x$. 
We have $\gamma(H)=\gamma(G)-1$, while the number of cut-vertices in $H$ was decreased by 2. Thus,
$$c(G)-2=c(H)\leq 2\gamma(H)=2(\gamma(G)-1).$$
Hence,
$$c(G)\leq 2\gamma(G).$$ 
So, we can assume that $z$ is included in at least two cliques, excluding $Q$. In this case we need to ensure whether the vertex $z$ belongs to all dominating sets of $G$ of size $\gamma(G)$. 
First, let $z$ belong to all dominating sets of $G$ of size $\gamma(G)$. Then, let $H$ be the graph obtained from $G$ by deleting $Q$, excluding $z$, with all pendant cliques containing $x$. Similarly to the previous case, we can assume that $x$ belongs to any dominating set of $G$ of size $\gamma(G)$. So we have $\gamma(H)=\gamma(G)-1$, while the number of cut-vertices in $H$ was decreased by 1. Thus,
$$ c(G)-1=c(H)\leq 2\gamma(H)=2(\gamma(G)-1).$$
Hence,
$$ c(G)\leq 2\gamma(G)-1\leq 2\gamma(G).$$
Thus, we are left with the case when $z$ does not belong to all dominating sets of $G$ of size $\gamma(G)$. First note, that $z$ belongs to at most one clique of level 1, otherwise $z$ would belong to all dominating sets of $G$ of size $\gamma(G)$.
Moreover, if there is a clique of level 1 that includes $z$, let us name it $Q^1_z$, one of its vertices must belong to a dominating set of $G$ of size $\gamma(G)$, so w.l.o.g. we can assume that $z$ is a vertex of $Q^1_z$ that belongs to a dominating set of $G$ of size $\gamma(G)$ and repeat the argument for the case where $z$ belongs to all dominating sets of $G$ of size $\gamma(G)$. Thus, we can assume that $z$ is not included in cliques of level 1. Let $z$ be included in a clique of level 2, let us name it $Q^*$, different from $Q$. Due to our assumption, $Q^*$ contains exactly two cut-vertices: the vertex $z$, and let us say the vertex $x'$. Note that we can assume that $x$ and $x'$ both belong to all dominating sets of $G$ of size $\gamma(G)$ and they both dominate $z$. Let $H$ be the graph obtained from $G$ by deleting $Q$, excluding $z$, and all pendant cliques that contain $x$. Since $x'$ dominates $z$ in $H$, we have $\gamma(H)=\gamma(G)-1$, while the number of cut-vertices was decreased by 1. Thus, $$c(G)-1= c(H)\leq 2\gamma(H)=2(\gamma(G)-1).$$
Hence,
$$ c(G)\leq 2\gamma(G)-1\leq 2\gamma(G).$$

So, we can assume that our vertex $z$ of the clique $Q$ is contained in at least two cliques of level at least 3. Let us show that such a situation is impossible for all cliques of level 2 in $G$. 

Assume the opposite. Take a clique $Q_0$ of level 2 and let $z_0$ be its unique vertex that may lie in cliques of higher level. By our assumption, $z_0$ lies in two cliques of level at least 3. Let $R_1$ be one of them, and let $r_1$ be the level of $R_1$. We have $r_1\geq 3$. By the definition of the level of a clique (see Subsection \ref{subsec:GraphTheoryConcepts}), there is a sequence of cliques of $G$, thanks to which $R_1$ got its label $r_1\geq 3$. Let $Q_1$ and $J_1$ be the cliques from this sequence with levels 2 and 3, respectively. Note that it may be the case that $R_1=J_1$. Let $z_1$ be the unique cut-vertex of $Q_1$ that may lie in cliques of higher level. Observe that $z_0\neq z_1$. By our assumption, $z_1$ lies in two cliques of level at least 3. Choose $R_2$ from them so that it is different from $J_1$. Let $r_2$ be the level of $R_2$. We have $r_2\geq 3$. By the definition of the level of a clique (see Subsection \ref{subsec:GraphTheoryConcepts}), there is a sequence of cliques of $G$, thanks to which $R_2$ got its label $r_2\geq 3$. Let $Q_2$ and $J_2$ be the cliques from this sequence with levels 2 and 3, respectively. Note that it may be the case that $R_2=J_2$. Let $z_2$ be the unique cut-vertex of $Q_2$ that may lie in cliques of higher level. Observe that $z_2\neq z_0$ and $z_2\neq z_1$. By our assumption, $z_2$ lies in two cliques of level at least 3. Choose $R_3$ from them so that it is different from $J_2$. Let $r_3$ be the level of $R_3$. We have $r_3\geq 3$. By the definition of the level of a clique (see Subsection \ref{subsec:GraphTheoryConcepts}), there is a sequence of cliques of $G$, thanks to which $R_3$ got its label $r_3\geq 3$. Let $Q_3$ and $J_3$ be the cliques from this sequence with levels 2 and 3, respectively. Note that it may be the case that $R_3=J_3$. Let $z_3$ be the unique cut-vertex of $Q_3$ that may lie in cliques of higher level. Observe that $z_3\neq z_0$, $z_3\neq z_1$ and $z_3\neq z_2$. We can continue this reasoning infinitely. Thus, our block graph must contain infinite number of vertices, which is not the case.

Thus, the situation described above is impossible for all cliques of level 2 in a finite graph $G$. This means that there exists a clique of level 2 in $G$ fulfilling one of conditions depicted earlier. The proof is complete.

\end{proof}



Note that we have proved the following inequalities.

$$(|V(G)|-MaxLeaf(G))/2\leq \gamma(G) \leq |V(G)|-MaxLeaf(G)$$

Thus, we have that $|V(G)|-MaxLeaf(G)$ and $\gamma(G)$ are equivalent from the perspective of FPT. Due to Proposition \ref{prop:n-max} we have the following.
\begin{proposition}
\textsc{Equitable Coloring} in block graphs is FPT with respect to the domination number. \label{prop:dom}
\end{proposition}
 Note that for general graphs, \textsc{Equitable Coloring} is paraNP-hard, when parameterized by minimum dominating set (cf. \cite{gomes:structural}). Our result shows that \textsc{Equitable Coloring} is much easier for block graphs, with respect to this parameter.

\section{\textsc{Equitable Coloring} and independent sets of block graphs}\label{sec:independence}

In this section, we consider block graphs and the \textsc{Equitable Coloring} problem from the perspective of independent sets of vertices and edges. In Subsection \ref{sec:alpha}, we work mainly with the parameter $\alpha_{\min}$ which has tight connections with the size of the largest independent set of vertices of a block graph. In Subsection \ref{sec:MatchingSection}, our focus is on matchings of block graphs. In particular, we view the problem from the angle of the number of vertices that a maximum matching of a block graph does not cover.

\subsection{Independent sets and the parameter $\alpha_{\min}$}\label{sec:alpha}

Recall that for a graph $G$, $\alpha(G)$ denotes the size of a largest independent vertex set in $G$. Let $\alpha(G,v)$ be the size of a largest independent set of $G$ that contains the vertex $v$. Define:
\[\alpha_{\min}(G)=\min_{v\in V(G)}\alpha(G,v).\]

Note that the parameter $\alpha_{\min}(G)$ is closely related to the topic of dominating sets in $G$. An \emph{independent} dominating set is such a dominating set $D$ that is independent. The \emph{independent} domination number of a graph $G$, denoted by $i(G)$, is the size of a smallest dominating set that is an independent set. Equivalently, it is the size of the smallest maximal independent set. Since every maximal independent set of size $\alpha_{\min}(G)$ is an independent dominating set, we have $\gamma(G) \leq i(G) \leq  \alpha_{\min}(G)$ for all graphs $G$. Thus, Proposition \ref{prop:dom}, combined with Lemma \ref{lem:Redparam}, implies the following results.

\begin{proposition}
\textsc{Equitable Coloring} of block graphs is FPT when parameterized by the independent domination number. \label{prop:i(G)}
\end{proposition}

\begin{proposition}
\textsc{Equitable Coloring} of block graphs is FPT when parameterized by $\alpha_{\min}$.\label{prop:alpha}
\end{proposition}


\subsection{Matchings in block graphs}\label{sec:MatchingSection}

In this section, we prove that EQUITABLE COLORING is hard in block graphs with a perfect matching.
Recall that in any graph $G$, we have
\[\nu(G)\leq \tau(G)\leq 2\cdot \nu(G).\]
Thus, the parameterization with respect to the vertex cover of $G$, $\tau(G)$, is equivalent to that of with respect to its matching number, $\nu(G)$. Theorem \ref{thm:fiala:cover} and Lemma \ref{lem:Redparam} imply

\begin{corollary}\label{cor:matchingFPT}
EQUITABLE COLORING is FPT when parameterized by the matching number.
\end{corollary}

\noindent One can try to strengthen this result. Since in any graph $$\tau(G)-\nu(G)\leq \nu(G),$$ we can ask about the parameterization with respect to $\tau(G)-\nu(G)$. \textsc{Equitable coloring} is NP-hard for bipartite graphs (cf. \cite{bod:part}). In these graphs, the difference $\tau(G)-\nu(G)$ is zero, thus the problem is paraNP-hard with respect to $\tau(G)-\nu(G)$, and hence unlikely to be FPT with respect to it. Observe that in any graph $G$, 
\[\nu(G)\leq |V|-\nu(G).\]
From Corollary \ref{cor:matchingFPT} and Lemma \ref{lem:Redparam}, we have that \textsc{Equitable Coloring} is FPT with respect to $|V|-\nu(G)$.
Thus one can try to do the next step trying to show that it is FPT with respect to $|V|-2\nu(G)$. 
We consider the restriction of the problem to block graphs with a perfect matching, i.e. with $|V|-2\nu(G)=0$.

Below we observe that \textsc{Equitable Coloring} is NP-hard for graphs containing a perfect matching. In order to demonstrate this, we will need a result by 
\cite{Sumner:1974}.
\begin{theorem}[\cite{Sumner:1974}]
\label{thm:Sumner} Let $G$ be a connected, claw-free graph on even number of vertices. Then $G$ has a perfect matching.
\end{theorem}

\begin{observation}
\label{obs:LineGraphsClawFree} Every line graph is claw-free.
\end{observation}

A classical result by 
\cite{Holyer:1981} states that the problem of testing a given bridgeless cubic graph for {3-edge-colorability} is NP-complete. One can always assume that the bridgeless cubic graph in this problem contains even number of edges. For otherwise, just replace one vertex with a triangle. The resulting graph is a bridgeless cubic graph on even number of edges and it is 3-edge-colorable if and only if the original graph is 3-edge-colorable. 

\begin{observation}
\label{obs:EqColoringNPhardPerfectMatchings} \textsc{Equitable Coloring} is NP-hard for 4-regular graphs containing a perfect matching.
\end{observation}

\begin{proof} We start with the \textsc{3-Edge-Coloring} problem for connected bridgeless cubic graphs with even number of edges. For such a graph $G$, consider its line graph $L(G)$. Observe that $G$ is 3-edge-colorable if and only if $L(G)$ is 3-vertex-colorable. Moreover, since in any 3-edge-coloring of $G$, the color classes must form a perfect matching, we have that the color classes in $G$ have equal size. Thus, the color classes in any 3-vertex-coloring of $L(G)$ must have equal size, too. Hence, $G$ is 3-edge-colorable if and only if $\chi_{=}(L(G))=3$.

Now, observe that $L(G)$ is connected, since $G$ is connected. Moreover, it is 4-regular, since $G$ is cubic. Finally, $|V(L(G))|$ is even since $G$ has even number of edges. Thus, $L(G)$ has a perfect matching via Theorem \ref{thm:Sumner}. The proof is complete.
\end{proof}

\begin{corollary}
\textsc{Equitable coloring} is paraNP-hard with respect to $|V|-2\nu(G)$ in the class of 4-regular graphs.
\end{corollary} 

Now, we are going to show that the problem remains hard even in block graphs with a perfect matching. In \cite{EqParamFPT}, some results are obtained about the parameterized complexity of the \textsc{Equitable Coloring} problem in block graphs. We will use them in order to obtain some further results. The \textsc{Bin Packing} problem is defined as follows: given a set of natural numbers $A=\{a_1,...,a_n\}$, two natural numbers $k$ and $B$, the goal is to check whether $A$ can be partitioned into $k$ parts such that the sum of numbers in each part is exactly $B$. In \cite{UnaryBinPacking}, it is shown that \textsc{Bin Packing} remains $W[1]$-hard with respect to $k$ even when the numbers are represented in unary. Below we prove the following

\begin{observation}
\label{obs:BinPackingParitykfixed} \textsc{Bin Packing} remains $W[1]$-hard with respect to $k$ even when the parity of $k$ is fixed.
\end{observation}

\begin{proof} We reduce \textsc{Bin Packing} to \textsc{Bin Packing} with fixed parity of $k$. Let $I$ be an instance of \textsc{Bin Packing}. If we are happy with the parity of $k$ in $I$, then we output the same instance. Assume that we are unhappy with the parity of $k$. Then consider the instance $I'=(A',k',B')$ defined as follows:
\[A'=A\cup \{1, B-1\}, k'=k+1, B'=B.\]
Observe that $I'$ can be constructed from $I$ in polynomial time. Moreover, $k'$ has a different parity than $k$ in $I'$. Let us show that $I$ is a yes-instance if and only if $I'$ is a yes-instance. If $I$ is a yes-instance, then we can add $\{1, B-1\}$ as a new bin and we will have a $(k+1)$ partition of $A'$ such that the sum in each partition is $B$. Now, assume that $I'$ is a yes-instance. Observe that $1$ and $B-1$ must be in the same bin and no other number can be with them. 
Thus the remaining $k$ sets in the partition form a $k$-partition in $A$.
Thus, $I$ is a yes-instance. The proof is complete.
\end{proof}

\begin{corollary}
\label{cor:Blockgraphoddk} \textsc{Equitable Coloring} remains $W[1]$-hard with respect to $k$ (the number of colors) in block graphs with odd values of $k$.
\end{corollary}

\begin{proof}In \cite{EqParamFPT}, the authors reduce the instance $I=(A,k,B)$ of (unary) \textsc{Bin Packing} to the equitable $(k+1)$-colorability of a block graph $G_I$. By Observation \ref{obs:BinPackingParitykfixed}, we can apply the same reduction only to instances of (unary) \textsc{Bin Packing} when $k$ is even. Thus, we will have that $(k+1)$ is odd for the resulting instances of \textsc{Equitable Coloring} in block graphs. The proof is complete.
\end{proof}

\begin{observation}
\label{obs:Eqblockgraphswithperfectmatching} \textsc{Equitable Coloring} remains $W[1]$-hard with respect to $k$ in block graphs with a perfect matching. 
\end{observation}

\begin{proof} Let us start with any instance of \textsc{Equitable Coloring} in block graphs where $k$ is odd. We will construct a graph $G'$ being an instance for the same  problem, but $G'$ will have a perfect matching such that $G$ has equitable $k$-coloring if and only if $G'$ does.

If $G$ has a perfect matching, then $G':=G$ and we are done. So, suppose $G$ does not have a perfect matching. Let $M(G)$ be a matching of the largest size in $G$ and  $V_M(G)=\{v\in V(G): v \not \in e,$ for any $e\in M(G)\}$, i.e. $V_M(G)$ is a subset of vertices of $G$ such that they are not end vertices of any edge belonging to $M(G)$, they are not covered by $M$. Recall that $k$ is odd. We construct $G'$ from $G$ by adding to every vertex of $G$ not covered by $M(G)$ an edge with a pendant clique $K_k$ (cf. Figure \ref{fig:gadgets}).

Since in every $k$-coloring of the added gadgets every color is used exactly the same number of times, then we immediately get the equivalence of equitable $k$-coloring of $G$ and $G'$. The proof is complete. 
\begin{figure}[htb]
    \centering
     \begin{tikzpicture}
  
  \node at (1.35,1) {$v$};
  \node at (0,1.5) {$G$};    
  \tikzstyle{every node}=[circle, draw, fill=black!50,
                        inner sep=0pt, minimum width=4pt]
     \node[circle,fill=black,draw] at (1,1) (n11) {};                     
    \node[circle,fill=black,draw] at (1,0) (n10) {}; 
    \node[circle,fill=black,draw] at (0,-1) (n0m1) {}; 
     \node[circle,fill=black,draw] at (2,-1) (n2m1) {}; 
    
    \path[every node]
            
            (n11) edge (n10)
            (n10) edge (n0m1)
            (n0m1) edge (n2m1)
            (n2m1) edge (n10);
            
            \draw[dashed] (-0.5,1.5) .. controls (-0.5,0.5) and (2,0.5) .. (2.5,1.5);
        
\end{tikzpicture}
    \caption{The construction of $G'$ from $G$ (cf. the proof of Observation \ref{obs:Eqblockgraphswithperfectmatching}): we add an edge with a clique of size $k$ to every vertex of $G$ not covered by the maxmimum matching. In this example, $k=3$.}
    \label{fig:gadgets}
\end{figure}
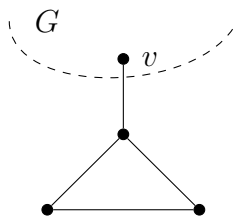
\end{proof}

\begin{corollary}
\label{cor:Eqblockgraphswithperfectmatching} \textsc{Equitable Coloring} is $W[1]$-hard with respect to $k+(|V|-2\nu(G))$ in block graphs. 
\end{corollary}

\begin{corollary}
If FPT$\neq W[1]$ then \textsc{Equitable Coloring} in block graphs is not FPT with respect to $|V|-2\nu(G)$.
\end{corollary}

\begin{proof}
If \textsc{Equitable Coloring} in block graphs were FPT with respect to $|V|-2\nu(G)$ then it would be polynomial time solvable for block graphs with a perfect matching. Hence it would be FPT with respect to $k$ for block graphs containing a perfect matching. By the previous observation, we know that it is W[1]-hard with respect to $k$ for block graphs with a perfect matching. Contradiction.
\end{proof}

\section{Equitable coloring and other structural parameters}
\label{sec:otherparameters}



In this section we try to strengthen the parameterized complexity results of \textsc{Equitable Coloring} problem by giving some new relations between structural parameters of block graphs. First, we need some auxiliaries.

\begin{proposition}
    \label{prop:SimplicalVertexAlpha} Let $G$ be a block graph and let $w$ be a simplicial vertex. Then $\alpha(G,w)=\alpha(G)$.
\end{proposition}
\begin{proof} Let $I$ be an independent set of $G$ of size $\alpha(G)$. If $w\in I$ then we are done. Thus, we can assume that $w\notin I$. Hence, there is a vertex $u\in I$ that lies in the unique clique $Q$ containing $w$. Consider the set $I'$ obtained from $I$ by replacing $u$ with $w$. Observe that $I'$ is an independent set of size $\alpha(G)$ and it contains $w$. The proof is complete. 
\end{proof}

\begin{lemma}
\label{lem:CutVertexSimplVertex} Let $G$ be a block graph, $v$ be a cut-vertex, and let $w$ be any simplicial vertex of $G$. Then $\alpha(G,v)\leq \alpha(G,w)$.
\end{lemma}
\begin{proof} The statement follows directly from Proposition \ref{prop:SimplicalVertexAlpha}. 
\end{proof}

\noindent The lemma implies
\begin{corollary}
\label{cor:CutVertexAlphaMin} For any block graph $G$ containing a cut-vertex, there is a cut-vertex $v$ such that $\alpha(G,v)=\alpha_{\min}(G)$.
\end{corollary}

\begin{proof} If $\alpha_{\min}(G)=\alpha(G)$, there is nothing to prove. On the other hand, if $\alpha_{\min}(G)<\alpha(G)$, then Lemma \ref{lem:CutVertexSimplVertex} implies that the minimum of $\alpha(G,z)$, $z \in V(G)$, is attained on cut-vertices. The proof is complete.
\end{proof}

\begin{lemma}
\label{lem:AddCliqeVertex} Let $G$ be a block graph obtained from a block graph $H$ by adding a clique $K$ to a vertex $u$ of $H$. Then, for any vertex $v\neq u$ of $H$, we have $\alpha(G,v)\leq 1+\alpha(H,v)$.
\end{lemma}

\begin{proof}
If $I$ is an independent set of $G$ of size $\alpha(G,v)$ containing $v$ then clearly $I$ can contain at most one vertex of $K$. Thus, we consider the set $I$ minus this vertex. It is an independent set of size $\alpha(G,v)-1$ in $H$. Thus, $\alpha(G,v)-1\leq \alpha(H,v)$, or equivalently, $\alpha(G,v)\leq 1+ \alpha(H,v)$. The proof is complete.
\end{proof}

\begin{lemma}
\label{lem:SubgraphAlphaMin} Let $G$ and $H$ be block graphs, such that $H$ is an induced subgraph of $G$. Assume that $v\in V(H)$ and $\alpha(G, v)=\alpha_{\min}(G)$. Then $\alpha_{\min}(H)\leq \alpha_{\min}(G)$.
\end{lemma}

\begin{proof} Let $I$ be a largest independent set of $H$ containing $v$. Clearly, $I$ is an independent set in $G$. Hence,
\[\alpha(H,v)=|I|\leq \alpha(G, v),\]
therefore
\[\alpha_{\min}(H)\leq \alpha(H,v)\leq \alpha(G, v)=\alpha_{\min}(G).\] 
The proof is complete.
\end{proof}

\begin{lemma}
\label{lem:PendantCliqueAtleast3Vertices} Let $G$ be a block graph and let $J$ be a pendant clique in $G$ with at least 3 vertices. Consider the block graph $H$ obtained from $G$ by removing one of the simplicial vertices of $J$. Let this vertex be $y$. Then $rad(G)\leq rad(H)$.
\end{lemma}

\begin{proof} Let $u$ be a vertex of $H$ such that $rad(H)=\epsilon_H(u)$. We can assume that $u\notin (J-x)$, where $x$ is the cut-vertex contained in $J$. Thus $\epsilon_H(u)= \epsilon_G(u)$, as if $P$ is a shortest $u-y$ path in $H$, then there is a path of the same length that does not end in $y$. Therefore
\[rad(G)\leq \epsilon_G(u)=\epsilon_H(u)=rad(H).\]
The proof is complete.
\end{proof}

\begin{lemma}
\label{lem:TwoPendantCliques} Let $G$ be a connected block graph. If $G$ has two pendant cliques, $J$ and $K$, with a common cut-vertex $v$, and $J$ has a simplicial vertex $y$, then  $rad(G)\leq rad(G-y)$.
\end{lemma}

\begin{proof} Let $H=G-y$ and let $u$ be a vertex of $H$ such that $rad(H)=\epsilon_H(u)$. Then $\epsilon_G(u)\leq \epsilon_H(u)$, as if $P$ is a $u-y$ path in $G$ of length $\epsilon_G(u)$, then there is a path of the same length $\epsilon_G(u)$ that does not end in $y$. Hence,
\[rad(G)\leq \epsilon_G(u)\leq \epsilon_H(u)=rad(H).\]
The proof is complete. 
\end{proof}

\begin{lemma}
\label{lem:Level2cliqueSimplicialVertex} Let $G$ be a block graph. Assume $Q$ is a level 2 clique in $G$, such that all its vertices are either simplicial or they belong to exactly one pendant clique, excluding one unique vertex of $Q$ that may lie in other non-pendant cliques. Then for any simplicial vertex $y$ in $Q$, $$rad(G)\leq rad(G-y).$$
\end{lemma}

\begin{proof} Let $H=G-y$ and let $u$ be a vertex of $H$ such that $rad(H)=\epsilon_H(u)$. Then $\epsilon_G(u)\leq  \epsilon_H(u)$, as for any shortest $u-y$ path in $G$ there is a path of the same length $d(u,y)$ in $G$ that misses $y$. Hence,
\[rad(G)\leq \epsilon_G(u)\leq \epsilon_H(u)=rad(H).\]
The proof is complete.
\end{proof}

\begin{lemma}
\label{lem:Level2cliqesAtleast3Vertices} 
Let $G$ be a block graph. Assume that there is a level 2 clique $Q$ in $G$ of size at least 3, such that all vertices of $Q$ belong to exactly one pendant clique, excluding one unique vertex of $Q$, let us name it $z$, that may lie in other non-pendant cliques. Then define the graph $H$ as follows. Let $x$ be any cut-vertex in $Q$ different from $z$. Let $J$ be the pendant clique containing $x$. 
Define $H=G-(V(J)-x)$. Then $rad(G)\leq rad(H)$.
\end{lemma}

\begin{proof} Let $u$ be a vertex of $H$ such that $rad(H)=\epsilon_H(u)$. Then $\epsilon_G(u)\leq \epsilon_H(u)$, as for any vertex $w$ in $G$ and a shortest $u-w$ path passing through $x$ there is a vertex $w'$ in $G$ and a path of the same length $d(u,w)$ that misses $x$ and the vertices of $J$. Hence,
\[rad(G)\leq \epsilon_G(u)\leq \epsilon_H(u)=rad(H).\]
The proof is complete.
\end{proof}

\noindent Now, we can pass to the main results.
\begin{theorem}
\label{thm:radiusalphamin} Let $G=(V,E)$ be a connected block graph. Then $rad(G)\leq \alpha_{\min}(G)$.
\end{theorem}
\begin{proof} Our proof is by induction on $|V|$. Clearly, the theorem is true when $|V|\leq 2$. Now, let $G$ be a connected block graph with at least 3 vertices. If $G$ is a star of cliques then clearly 
\[rad(G)=\alpha_{\min}(G)=1.\]
Thus, the statement is trivial for this case. Hence, we can assume that $G$ is not a star of cliques. Let $v$ be a cut-vertex with $\alpha(G,v)=\alpha_{\min}(G)$ (cf. Corollary \ref{cor:CutVertexAlphaMin}). Since $v$ is a cut-vertex, it is contained in at least two cliques. 
Let us assume that there are two pendant cliques containing $v$, and let $J$ be one of them, while $K$ is the other one. Let $y$ be a simplicial vertex of $J$. Consider the block graph $H=G-y$. By Lemma \ref{lem:TwoPendantCliques}, we have $rad(G)\leq rad(H)$, and by Lemma \ref{lem:SubgraphAlphaMin}, $\alpha_{\min}(H)\leq \alpha_{\min}(G)$. 
%
Therefore,
\[rad(G)\leq rad(H)\leq \alpha_{\min}(H)\leq \alpha_{\min}(G).\]
Thus, we can further assume that there is at most one pendant clique around $v$. Similarly, we can assume that any other ($\neq v$) cut-vertex of $G$ is contained in at most one pendant clique.


Let $Q$ be a clique in $G$ of level 2. Observe that it contains at most one vertex $z$ that may be contained in another non-pendant clique. All other vertices of $Q$ are either simplicial or they are contained in exactly one pendant clique. 
If $y$ is a simplicial vertex in $Q$, then consider the graph $H=G-y$. Observe that $H$ is a block graph of order smaller than $G$. Hence, we have $rad(H)\leq \alpha_{\min}(H)$. By Lemma \ref{lem:Level2cliqueSimplicialVertex}, $rad(G)\leq rad(H)$, and by Lemma \ref{lem:SubgraphAlphaMin}, we have $\alpha_{\min}(H)\leq \alpha_{\min}(G)$.
%
Therefore, 
\[rad(G)\leq rad(H)\leq \alpha_{\min}(H)\leq \alpha_{\min}(G).\]
Thus, we can further assume that $Q$ contains no simplicial vertices. Hence all vertices of $Q$, except at most $z$, are contained in exactly one pendant clique. 
Now, let us assume that $|Q|\geq 3$. 
Since $Q$ is a clique of level 2, there is at least one non-simplicial vertex $x$ in $Q$, except $z$. 
Let $J$ be the pendant clique containing $x$. Define $H=G-(V(J)-x)$. Note that $H$ still contains vertex $v$ and it is a block graph of order smaller than $G$. Hence, we have $rad(H)\leq \alpha_{\min}(H)$. By Lemma \ref{lem:Level2cliqesAtleast3Vertices}, $rad(G)\leq rad(H)$, and by Lemma \ref{lem:SubgraphAlphaMin} $\alpha_{\min}(H)\leq \alpha_{\min}(G)$. 
Therefore,
\[rad(G)\leq rad(H)\leq \alpha_{\min}(H)\leq \alpha_{\min}(G).\]
Thus, we can focus on the remaining case where $Q=K_2$. Let $J$ be the unique clique containing the other ($\neq z$) vertex $x$ of $Q$. If $|J|\geq 3$, then let $y$ be a simplicial vertex in $J$. Consider the graph $H=G-y$. It still contains the vertex $v$ and it is a block graph of order smaller than $G$. Hence we have $rad(H)\leq \alpha_{\min}(H)$. By Lemma \ref{lem:PendantCliqueAtleast3Vertices}, $rad(G)\leq rad(H)$, and by Lemma \ref{lem:SubgraphAlphaMin} $\alpha_{\min}(H)\leq \alpha_{\min}(G)$.
%
Therefore, 
\[rad(G)\leq rad(H)\leq \alpha_{\min}(H)\leq \alpha_{\min}(G).\]
If $J=K_2$ then observe that this conclusion holds for every clique $Q$ chosen as above. 

Now, consider the graph $H$ containing $v$ obtained from $G$ by removing the vertices of all pendant cliques $K_2$ except the one around $v$ (if it exists). We remove the two vertices of $K_2$ for each choice of such $K_2$'s. Observe that $rad(G)\leq rad(H)+2$. 
In order to see this, let us observe that,
by the definition of $rad(H)$, we have that for some vertex $w\in V(H)$, $\epsilon_{H}(w)=rad(H)$.
Now, by construction,
\[rad(G)\leq \epsilon_{G}(w)\leq \epsilon_{H}(w)+2=rad(H)+2.\]
If $\alpha(G,v)\geq \alpha(H,v)+2$ then
\[rad(G)\leq rad(H)+2\leq \alpha_{\min}(H)+2\leq \alpha(H,v)+2\leq \alpha(G,v)=\alpha_{\min}(G).\]
Thus, we can assume that $\alpha(G,v)\leq \alpha(H,v)+1$. In particular, this means that we have at most one choice for $J$ above. Moreover, there is at most one pendant clique in $G$ that does not contain $v$. Let $T$ be any clique of $G$. Let us show that the clique-degree of $T$ is at most 2. Assume that $T$ has clique degree at least 3. If $v$ does not lie in $T$, then $G$ contains at least three pendant cliques. Note that at least two of them will not contain $v$ contradicting our conclusion above that there should be at most one such a pendant clique. So assume $v$ lies in $T$. Since the clique degree of $T$ is at least three, there will be two pendant cliques that will not contain $v$. Again, this contradicts our conclusion above that this number should be at most one. Thus, any clique of $G$ has clique degree at most 2. Since $G$ is connected, we have that $G$ is a path of cliques. 


Now, we show that $G$ is a path. It suffices to show that $G$ has no simplicial vertices in its internal cliques. If we assume that there is such a simplicial vertex $y$, then consider the graph $H=G-y$ containing $v$. Observe that $H$ is a block graph of order smaller than $G$. Hence we have $rad(H)\leq \alpha_{\min}(H)$. Let $u$ be a vertex of $H$ such that $rad(H)=\epsilon_H(u)$. Then $\epsilon_G(u)\leq \epsilon_H(u)$, as for any shortest $u-y$ path of $G$ there is a path of the same length $d(u,y)$ that misses $y$. Hence,
\[rad(G)\leq \epsilon_G(u)\leq \epsilon_H(u)=rad(H).\]
Lemma \ref{lem:SubgraphAlphaMin} implies $\alpha_{\min}(H)\leq \alpha_{\min}(G)$. Therefore, 
\[rad(G)\leq rad(H)\leq \alpha_{\min}(H)\leq \alpha_{\min}(G).\]
Thus, let we are left with the case when $G$ is a path on $|V(G)|=n$ vertices. A direct check shows that no path with $n\leq 5$ is a counter-example to our statement. Thus, $n\geq 6$. Since $rad(G)=\lfloor \frac{n}{2} \rfloor$, it suffices to show that $\alpha(G, v)\geq \lfloor \frac{n}{2} \rfloor$. Since $n\geq 6$, we can find a degree-one vertex $y$, such that the path $H=G-y-z$ contains $v$. Here $z$ is the unique neighbor of $y$. Since $H$ is a path of order $n-2$, we have
\[\alpha(H, v)\geq \alpha_{\min}(H)\geq rad(H)= \left \lfloor \frac{n-2}{2} \right \rfloor.\]
Moreover, $\alpha(G, v)\geq \alpha(H, v)+1$, as if $I$ is a largest independent set of vertices in $H$ that contains $v$, then we can always get a similar set in $G$ just by adding $y$ to $I$. Hence,
\[rad(G)=\left \lfloor \frac{n}{2} \right \rfloor=\left \lfloor \frac{n-2}{2} \right \rfloor +1\leq \alpha(H, v)+1\leq \alpha(G, v)=\alpha_{\min}(G).\]
The proof is complete. 
\end{proof}

In the next theorem we bound the radius of a block graph by a function of its $dc(G)$. We precede the theorem with some simple observations concerning block graphs.

\begin{observation}
\label{obs:Shortestpathstructure} For any vertices $u$ and $v$ of $G$, the internal vertices of any shortest $(u-v)$-path are cut-vertices.
\end{observation}

\begin{observation}
\label{obs:EccentricityPath} Let $u$ be a vertex in $G$ and let $P$ be a $(u-v)$-path such that $P$ is of length $\epsilon_G(u)$. Then all internal vertices of $P$ are cut-vertices and $v$ is a simplicial vertex.
\end{observation}

\begin{observation}
\label{obs:Induced2path} Assume that a vertex $y$ is adjacent to vertices $x$ and $z$ such that $xz\notin E$. Then for any $dc$-set $D$, we have $D\cap \{x,y,z\}\neq \emptyset$.
\end{observation}

\begin{observation}
\label{obs:dcmonotonicity} Let $H$ and $G$ be two graphs with $H\subseteq G$. Then $dc(H)\leq dc(G)$.
\end{observation}

Note that the difference $dc(G)- rad(G)$ can be arbitrarily large. In order to see this, let $G$ be a graph obtained from a star of at least two cliques $K_3$, sharing vertex $v$, by adding one clique $Q=K_{k+2}$ to one of simplicial vertices of $K_3$ in the star. The common vertex of $Q$ and the star of cliques is named by $x$ (Fig. \ref{fig:exdcrad}). Finally, we add exactly one pendant clique, of size at least 3, to each simplicial vertex of $Q$. Note that $rad(G)=2$ - the center is formed by vertex $x$, while $dc(G)=k+2$. Any $dc$-set $D$ of $G$ of size $dc(G)$ is formed by vertices $v$, $x$ and $k$ cut-vertices of $Q$, excluding $x$ (cf. Fig. \ref{fig:exdcrad}).

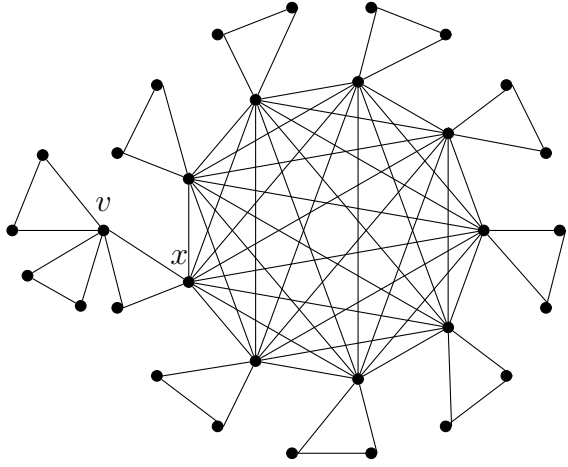
\begin{figure}[htb]
\centering
\begin{tikzpicture}
\node at (-2,-0.35) {$x$};
\node at (-3,0.35) {$v$};
\tikzstyle{every node}=[circle, draw, fill=black!50, inner sep=0pt, minimum width=4pt]
\node[circle,fill=black,draw] at (-3.8,1) (d1) {};
\node[circle,fill=black,draw] at (-4.2,0) (d2) {};
\node[circle,fill=black,draw] at (-4,-0.6) (d3) {};
\node[circle,fill=black,draw] at (-3.3,-1) (d4) {};
\foreach \a in {1,2,...,9}{
\draw (\a*360/9: 2cm)
node[circle,fill=black,draw](i\a) {};
}
\foreach \a in {1,2,...,18}{
\draw (\a*360/18: 3cm)
node[circle,fill=black,draw](o\a) {};
}

\foreach \a in {1,2,...,8}{
 \foreach \b in {\a, ...,9}
 {\draw (i\a) -- (i\b);}
}

\foreach \a in {1,3,...,17}{
\tikzmath{\b=\a+1;}
  \draw (o\a) -- (o\b);
}  
\foreach \a in {1,2,...,9}{
\tikzmath{\b= 2 * (\a -1) +1; \c=\b+1;}
\draw (i\a) -- (o\b);
\draw (i\a) -- (o\c);
}

\draw (o9) -- (d1);
\draw (o9) -- (d2);
\draw (o9) -- (d3);
\draw (o9) -- (d4);
\draw (d1) -- (d2);
\draw (d3) -- (d4);
\end{tikzpicture}
\caption{An exemplary block graph $G$ with $dc(G)-rad(G)=k$ with $k=7$.}\label{fig:exdcrad}
\end{figure}

\begin{theorem}
\label{thm:radiusdc(G)2/3} Let $G=(V,E)$ be a connected block graph. Then $rad(G)\leq \frac{3}{2}\cdot dc(G)+1$.
\end{theorem}

\begin{proof} 
Our proof is by induction on $|V|$. Clearly, if $|V|\leq 2$, then $dc(G)=0$ and $rad(G)\leq 1$, and the statement holds. Now, let us assume that our statement holds for every block graph $G$ with $|V|<k$, $k\geq 3$. Now, let $G$ be a block graph with at least $k$ vertices, $k\geq 3$. If $G$ is a star of cliques, then $rad(G)=1$ and the statement is true, independly on the value of $dc(G)$. Hence, we can assume that $G$ is not a star of cliques, i.e. $G$ has at least two cut-vertices.

If there is a non-pendant clique that contains a simplicial vertex $y$ then consider the block graph $H=G-y$. As in Lemma \ref{lem:Level2cliqueSimplicialVertex}, $rad(G)\leq rad(H)$. Hence,
\[rad(G)\leq rad(H)\leq 1+\frac{3}{2}\cdot dc(H)\leq 1+\frac{3}{2}\cdot dc(G).\]

Thus, we can assume now that all non-pendant cliques (that is, all cliques of level at least two) contain only cut-vertices. If there is a pendant clique with at least two simplicial vertices, then consider the block graph $H$ obtained from $G$ by removing one of the simplicial vertices, let name it by $y$.
By Lemma \ref{lem:PendantCliqueAtleast3Vertices}, $rad(G)\leq rad(H)$.
Hence,
\[rad(G)\leq rad(H)\leq 1+\frac{3}{2}\cdot dc(H)\leq 1+\frac{3}{2}\cdot dc(G).\]

Next, we can assume that each pendant clique contains exactly one simplicial vertex. In particular, this means that all pendant cliques in $G$ are isomorphic to $K_2$.
If there is a cut-vertex $w$ in $G$ that is contained in at least two pendant cliques, let us say $J_1=\{w,z_1\}$ and $J_2=\{w,z_2\}$, where $z_1$ and $z_2$ are simplicial vertices, then consider the graph $H=G-z_2$. By Lemma \ref{lem:TwoPendantCliques}, $rad(G)\leq rad(H)$.
Hence,
\[rad(G)\leq rad(H)\leq 1+\frac{3}{2}\cdot dc(H)\leq 1+\frac{3}{2}\cdot dc(G).\]

Thus, we can also assume that each cut-vertex of $G$ is contained in at most one pendant clique. Let $Q$ be a clique of level exactly two. Observe that it has at most one vertex $z$ that is contained in another non-pendant clique. Note that there is at least one cut-vertex $x$ in $Q$, except $z$. Since $Q$ contains no simplicial vertices, we have that all vertices of $Q$, maybe except $z$, are contained in exactly one pendant clique. 
If $|Q|\geq 3$, then let $x$ be any cut-vertex in $Q$ different from $z$. Let $J=\{x, x_J\}$ be the pendant clique containing $x$. Define $H=G-x_J$. 
By Lemma \ref{lem:Level2cliqesAtleast3Vertices}, $rad(G)\leq rad(H)$.
Hence,
\[rad(G)\leq rad(H)\leq 1+\frac{3}{2}\cdot dc(H)\leq 1+\frac{3}{2}\cdot dc(G).\]
Thus, we can assume that $Q=K_2$. 

If $G$ contains no level 3 clique, then it is easy to see that $rad(G)=2$ and $dc(G)=1$. Thus, our inequality is true for this case. This means, that we can assume that $G$ contains at least one level 3 clique $R$. If $R$ is a clique of size at least 3, we use the following reasoning. Note, since $R$ is a clique of level 3, by the definition, there is at most one cut-vertex $z_R$ of $R$ that may be contained in cliques of level at least 4. The rest of vertices (that is, all except $z_R$), which are cut-vertices (internal cliques do not contain simplicial vertices), are contained in a level 1 or a level 2 clique. Moreover, since $R$ is of level 3, at least one of its cut-vertices is contained in a level 2 clique $Q$ which in turn is adjacent to a pendant clique, by the definition of $Q$.
If one of these cut-vertices $w\neq z_R$ of $R$ is contained in a pendant clique $K_2=wz$, then consider the graph $H=G-z$. Let $u$ be a vertex of $H$ such that $rad(H)=\epsilon_H(u)$. Then $\epsilon_H(u)= \epsilon_G(u)$, as for any shortest $u-y$ path of $G$ there is a path of the same length that misses $y$. Hence,
\[rad(G)\leq \epsilon_G(u)=\epsilon_H(u)=rad(H),\]
and, therefore,
\[rad(G)\leq rad(H)\leq 1+\frac{3}{2}\cdot dc(H)\leq 1+\frac{3}{2}\cdot dc(G).\]
Next, if one of these cut-vertices $w$ (that is, anyone except $z_R$) is contained in a level 2 clique $Q'=wz$ and $z$ is contained in a pendant clique $J'=zy$, then consider the graph $H=G-z-y$. Let $u$ be a vertex of $H$ such that $rad(H)=\epsilon_H(u)$. We can assume that $u\notin J$. Thus $\epsilon_H(u)= \epsilon_G(u)$. Hence,
\[rad(G)\leq \epsilon_G(u)=\epsilon_H(u)=rad(H),\]
and, therefore,
\[rad(G)\leq rad(H)\leq 1+\frac{3}{2}\cdot dc(H)\leq 1+\frac{3}{2}\cdot dc(G).\]
Thus, we can assume that the clique $R$ is $K_2$, in particular we can assume that any level 3 clique is isomorphic to $K_2$. Moreover, the cut-vertex $x$, that belongs to both $Q$ and $R$, is of degree two.

Now, let $G$ contain at least two level 3 cliques $R_1$ and $R_2$. 
Let cliques $Q_1, J_1$ and cliques $Q_2, J_2$ be the corresponding level 2 and level 1 cliques corresponding to cliques $R_1$ and $R_2$, respectively. We have
\[(V(J_1)\cup V(Q_1))\cap (V(J_2)\cup V(Q_2))=\emptyset.\]
By Observation \ref{obs:Induced2path}, any $dc(G)$ set $D$ intersects $V(J_1)\cup V(Q_1)$ and $V(J_2)\cup V(Q_2)$. Thus, if we define graph $H$ as $G-V(J_1)-V(Q_1)-V(J_2)-V(Q_2)$, then 
\[dc(H)\leq |D\cap V(H)|\leq |D|-2= dc(G)-2.\]
 Let $u$ be a vertex of $H$ such that $rad(H)=\epsilon_H(u)$. Since any vertex of $H$ is in distance at least 3 to a simplicial vertex of $G-V(H)$ then $\epsilon_G(u)\leq \epsilon_H(u)+3$. Hence,
\[rad(G)\leq \epsilon_G(u)\leq \epsilon_H(u)+3=rad(H)+3,\]
and, therefore,
\[rad(G)\leq rad(H)+3\leq 4+\frac{3}{2}\cdot dc(H)\leq 4+\frac{3}{2}\cdot (dc(G)-2)=1+\frac{3}{2}\cdot dc(G).\]

Thus, we are left with the case when there is exactly one level 3 clique $R$ in $G$. It is not hard to see that the graph $G$ is isomorphic to $P_6$ and the two vertices of $R$ form the center of $G$. Moreover, their eccentricity is 3 and $rad(G)=3$. On the other hand, $dc(G)= 2$. Hence, 
\[rad(G)=3\leq 1+\frac{3}{2}\cdot 2=1+\frac{3}{2}\cdot dc(G)\]
and the proof is complete.
\end{proof}

\begin{remark} The bound presented in the previous theorem is tight for infinitely many block graphs. Let $P_n$ be the path on $n$ vertices. Observe that
\[
    rad(P_n)= \left\lfloor \frac{n}{2} \right\rfloor.
\]
Using Observation \ref{obs:Induced2path}, it can be shown that
\[
dc(P_n)= \left\lfloor \frac{n}{3}\right\rfloor.
\]

Thus, for $n \equiv 2 \mod 6$, we will have $rad(P_n)=\frac{3}{2}\cdot dc(P_n)+1$.
\end{remark}

In Theorem \ref{thm:radiusalphamin}, we have shown that in any connected block graph $G$, $rad(G)\leq \alpha_{\min}(G)$. Thus, one can try to strengthen the result about the parameterization of \textsc{Equitable Coloring} with respect to $\alpha_{\min}(G)$, by showing that it is FPT with respect to $rad(G)$. Unfortunately, it turns out that such a result is unlikely to be true. In \cite{EqParamFPT}, it is shown that \textsc{Equitable Coloring} is W[1]-hard with respect to $diam(G)$ - the diameter of $G$ (cf. Theorem \ref{diam4}), for block graphs. Since in any graph $G$, not necessarily block graph,
\[rad(G)\leq diam(G)\leq 2\cdot rad(G),\]
from Lemma \ref{lem:Redparam}, we have that $diam(G)$ and $rad(G)$ are equivalent from the perspective of FPT. Thus, \cite{EqParamFPT} implies that equitable coloring is unlikely to be FPT with respect to $rad(G)$ even when the input is restricted to block graphs.



\section*{Conclusion}
\label{sec:conclusion}

In this paper, we discussed the problem of \textsc{Equitable Coloring} of block graphs with respect to many different parameters. Our research completes the approach given in \cite{iterated, fiala, gomes:structural, EqParamFPT}. We presented some graph theoretic results that relate various parameters in block graphs. We also discussed algorithmic implications of these results.

Many parameters still remain open for the problem of \textsc{Equitable Coloring}. 
\cite{gomes:structural} depicted as an open case the problem of \textsc{Equitable Coloring} with such parameters as feedback edge set and feedback vertex set with maximum degree of an arbitrary graph. Hence, the further considerations over \textsc{Equitable Coloring} with respect to different parameters is still desirable, both for general and particular graph classes.

\section*{Acknowledgement} We would like to thank our anonymous referees for very careful reading of the manuscript and many insightful comments and suggestions that helped us to improve the presentation of the paper.

\nocite{*}
\bibliographystyle{abbrvnat}

\end{document}